\documentclass[
reprint,
superscriptaddress,
nofootinbib,
amsmath,amssymb,
aps,
prl,
floatfix,longbibliography
]{revtex4-2}
\usepackage{graphicx}% Include figure files
\usepackage{dcolumn}% Align table columns on decimal point
\usepackage[hypertexnames, breaklinks=true, bookmarksnumbered=true, bookmarksopen=true, colorlinks=true, linktocpage=true, citecolor=magenta, urlcolor=magenta, linkcolor=magenta]{hyperref}% add hypertext capabilities
\usepackage{amsmath,amssymb,amsfonts,amsthm}
\usepackage{braket}
\usepackage{mathtools} %interesting stuff
\usepackage[T1]{fontenc}
\usepackage[utf8]{inputenc}
\usepackage{txfonts}
\usepackage{bbm} %Enables correct fonts for the fields like C,R etc.
\usepackage{bm} %Bold Math
\usepackage{mathrsfs} %For nice looking curly CAPS
\usepackage{xcolor}
\usepackage{natbib}
\usepackage{physics}
\allowdisplaybreaks

\DeclareMathOperator{\sgn}{sgn}

\newtheorem{theorem}{Theorem}
\newtheorem{lemma}{\emph{Lemma}}

\newtheorem*{conjecture}{Conjecture}
\newtheorem{fact}{Fact}
\begin{document}

\title{Coherence manipulation in asymmetry and thermodynamics}
\author{Tulja Varun Kondra$^\dagger$}
\email{tuljavarun@gmail.com}
\affiliation{Centre for Quantum Optical Technologies, Centre of New Technologies,
University of Warsaw, Banacha 2c, 02-097 Warsaw, Poland}
\affiliation{Institute for Theoretical Physics III, Heinrich Heine University D\"{u}sseldorf, Universit\"{a}tsstra{\ss}e 1, D-40225 D\"{u}sseldorf, Germany}

\author{Ray Ganardi$^\dagger$}
\affiliation{Centre for Quantum Optical Technologies, Centre of New Technologies,
University of Warsaw, Banacha 2c, 02-097 Warsaw, Poland}
\author{Alexander Streltsov}
\affiliation{Institute of Fundamental Technological Research, Polish Academy of Sciences, \\ Pawińskiego 5B, 02-106 Warsaw, Poland}
\affiliation{Centre for Quantum Optical Technologies, Centre of New Technologies,
University of Warsaw, Banacha 2c, 02-097 Warsaw, Poland}

\begin{abstract}
In the classical regime, thermodynamic state transformations are governed by the free energy.
This is also called as the second law of thermodynamics.
Previous works showed that, access to a catalytic system allows us to restore the second law in the quantum regime when we ignore coherence.
However, in the quantum regime, coherence and free energy are two independent resources.
Therefore, coherence places additional non-trivial restrictions on the the state transformations, that remains elusive.
In order to close this gap, we isolate and study the nature of coherence, i.e.\ we assume access to a source of free energy.
We show that allowing catalysis along with a source of free energy allows us to amplify any quantum coherence present in the quantum state arbitrarily.
Additionally, any correlations between the system and the catalyst can be suppressed arbitrarily.
Therefore, our results provide a key step in formulating a fully general law of quantum thermodynamics.
\end{abstract}

\maketitle
\def\thefootnote{$\dagger$}\footnotetext{These authors contributed equally to this work.}

Thermodynamics was motivated by the need to understand what are the allowed state transformations of a system.
The restrictions are codified into \emph{the laws of thermodynamics}, which take on a central importance in the way we understand how nature works.
One of the highlights is the concept of free energy, that is, the maximum available work that can be reversibly extracted in a cycle from a system.
Due to the time of its conception, this was done in a classical setting, with a large ensemble that is in equilibrium, and thermodynamic quantities understood as averages.

Recently, there has been an effort to understand how thermodynamics translates to quantum systems~\cite{janzing2000thermodynamic,Horodecki_2013}.
Some of the most significant advances on this front work within the so-called resource theory approach, with the Gibbs state as the free state and the set of thermal operations as the free operations~\cite{Chitambar_2019,Horodecki_2013,Lostaglio_2019,Ng2018}.
Notably, it was shown that transformations between classical (diagonal) states are actually governed by a ``family'' of free energies, that reduce to the standard Gibbs free energy in the macroscopic limit~\cite{Brandão_2015}.
On a technical level, they proved that these free energies characterize the allowed exact catalytic transformations, with the environment fulfilling the catalytic conditions \cite{datta2022catalysis}.
This makes sense since the state of the environment should be unchanged in a cyclic process.

However, requiring exact catalysis might often be too strict. An analogous result showed that if we allow for a small system-catalyst correlation (approximate catalysis), state transformations between classical states are completely characterised by the standard Gibbs free energy~\cite{PhysRevX.8.041051}. 
%This suggests that allowing for small correlations can completely change the laws of transformation.
%In addition, it was also known that approximate catalytic transformations between general quantum states under Gibbs-preserving operations are characterized by the so called non-equilibrium free energy defined by the quantum Kullback-Leibler (KL) divergence~\cite{PhysRevLett.126.150502}. 
%Note that the set of Gibbs-preserving operations is strictly larger than the set of thermal operations, due to coherent effects~\cite{faist2015gibbs}.
For transformations between general states, thermal operations additionally exhibit coherent restrictions.
Even when we augment thermal operations with an unbounded source of free energy, these restrictions due to coherence still persist, forbidding transformations that create coherence~\cite{Lostaglio_2015}.
In fact, this setting gives rise to the so-called resource theory of asymmetry \cite{Lostaglio_2015}.
Therefore, studying the resource theory of asymmetry allows us to understand the restrictions on coherence manipulation, imposed by the allowed thermodynamic transformations. While there are some recent results on characterizing catalytic transformations in this setting~\cite{PhysRevA.103.022403,Takagi_2022,PhysRevLett.123.020403,PhysRevLett.123.020404}, a general characterization is still missing.

In this work, we argue that the restrictions that arise solely from coherence can be expressed as a set-inclusion relation.
This is because any non-zero coherence in the initial state can be amplified by an arbitrary amount.
Our result answers a conjecture posed in Ref.~\cite{Takagi_2022} in the positive.
Moreover, it suggests that the laws of thermodynamics take on a simple form, determined by free energy and set-inclusion of the coherences.

\emph{Preliminaries.---}
Resource theories provide us with a mathematical framework to study various quantum resources \cite{Chitambar_2019}.
A (quantum) resource theory is defined by a set \emph{free states} $\mathcal{F}_s$ that do not contain any resources, and a set of \emph{free operations} $\mathcal{F}_o$ that are easily implementable. 
We require that the set of free states is preserved by the free operations, i.e. for every free state $\rho$ and every free operation $\Lambda$, we have $\Lambda(\rho)\,\in\, \mathcal{F}_s$.

This captures the intuition that \emph{resources} cannot be created for free.
Depending on the structure of the theory, having access to a (resource) state $\sigma\notin \mathcal{F}_s$ might allow us to access some other resource states as well as implement operations which are not free, using free operations.
This resource theoretic approach has be successfully applied to study various quantum resources such as entanglement~\cite{Horodecki_2009}, coherence~\cite{Streltsov_2017}, purity~\cite{HorodeckiPhysRevA.67.062104,GOUR_purity,Streltsov_2018}, imaginarity~\cite{Gour2018,Wu_PRL,IM2,varun_imaginarity} and magic \cite{veitch2014resource,PhysRevLett.118.090501} .

The \emph{resource theory of thermodynamics} was defined to model the thermodynamics of quantum systems \cite{Lostaglio_2019}. 
To every system ($S$), we associate a Hamiltonian $H_S$ and the Gibbs state $\gamma_{S}$, where $\gamma_{S}=\frac{e^{-\beta H_{S}}}{\Tr e^{-\beta H_{S}}}$. Here, $\beta=\frac{1}{k_BT}$ is the inverse temperature. The set of free states is defined to contain only the Gibbs state, motivated by the classic result that it is the only state from which we cannot extract any work, even from multiple copies~\cite{Pusz_1978}. 
When we have a joint system, we assume that the joint Hamiltonian is simply a sum of the local Hamiltonians, i.e.\ $H_{SS'} =H_S\otimes I_{S'}+I_S\otimes H_{S'}$.
The set of free operations are defined to be thermal operations \cite{Lostaglio_2019}.
These are the operations that can be built up from the following steps:
\begin{enumerate}
    \item Bringing the system in contact with another system in a Gibbs/thermal state, $\rho_S \mapsto \rho_S \otimes \gamma_{S'}$.
    \item Applying an energy-preserving unitary, $\rho_S\otimes \gamma_{S'} \mapsto U \rho_S\otimes \gamma_{S'} U^\dagger$, with $[U, H_{SS'}] = 0$.
    \item Tracing out subsystems, $\rho_{SS'} \mapsto \Tr_{S'} \rho_{SS'}$.
\end{enumerate}

There are two properties that together feature in thermal operations: \cite{Lostaglio_2015}
\begin{enumerate}
\item Time-translation invariance 
    \begin{equation}\label{time_trans_inv}
         \forall \rho, t,\,  e^{-iH_{S'}t}\Lambda_{S\to S'}\pqty{\rho}e^{iH_{S'}t}=\Lambda_{S\to S'}\pqty{e^{-iH_St}\rho e^{iH_St}}
    \end{equation}
\item Gibbs state preservation
    \begin{equation}\label{Gibbs_pres}
        \Lambda_{S\rightarrow S'}\pqty{\gamma_S}=\gamma_{S'}
    \end{equation}
\end{enumerate}
First condition alone gives us covariant operations, while keeping only the second gives us Gibbs-preserving operations.
Note that when we augment thermal operations with an infinite store of incoherent work, we obtain exactly the set of covariant operations.
Formally, adding $\rho_S \mapsto \rho_S \otimes \omega_{S'}$ to the set of thermal operations allow us to implement any covariant operations.
Here, we require that $\omega_{S'}$ is incoherent, i.e.\ $[\omega_{S'}, H_{S'}] = 0$. Recall that covariant operations are exactly the set of maps that are invariant under time-translation, and thus we obtain the resource theory of asymmetry~\cite{Marvian_2014}.
In contrast, a similar prescription to obtain Gibbs-preserving maps is an open problem~\cite{gibbs_open_problem}.

There is another crucial ingredient in thermodynamical processes, that is the environment.
In a cyclic process, we expect to be able to ``borrow'' a system, as long as we return it in the same state in the end.
This is exactly the notion of catalysis in state transformations, first introduced in the theory of entanglement~\cite{Jonathan_1999} and extended to other resource theories more recently~\cite{datta2022catalysis,lipkabartosik2023catalysis}.
In the context of thermodynamics, Ref.~\cite{Brandão_2015} showed that exact catalytic transformations between diagonal states are governed by a family of free energies,
i.e.\ for any incoherent states $\rho, \sigma$, the free energies $S_\alpha (\rho || \gamma) = \frac{\sgn{\alpha}}{\alpha - 1} \log \Tr \rho^\alpha \gamma^{1-\alpha}$ are non-decreasing $S_\alpha(\rho||\gamma)\geq S_\alpha(\sigma||\gamma)$ for all $\alpha \in(-\infty,\infty)$ if and only if there exists a catalyst $\tau_C$ and a thermal operation $\Lambda_{SC}$ such that $\Lambda_{SC} \pqty{\rho_S \otimes \tau_C} = \sigma_S \otimes \tau_C$. This is also known as an exact catalytic transformation. However, requiring the final state to be in exact tensor product might be too strict, as requiring the marginal state of the catalyst to be preserved and the system-catalyst correlations to be small is often a good enough approximation. Therefore, we will say $\rho$ can be transformed into $\sigma$ with approximate catalysis if for every $\varepsilon>0$ there is a catalyst $\tau_C$ and a thermal operation $\Lambda_{SC}$ such that 
\begin{equation}\label{approximate}
\begin{aligned}\mu_{SC} =\Lambda_{SC}\pqty{\rho_S\otimes\tau_C},\,\,
||\mu_{SC}-\sigma_S\otimes\tau_C||_{1} < \varepsilon\,\,\text{and}\,\mu_C=\tau_C.
\end{aligned}
\end{equation}
Here, $||M||_1 = \mathrm{Tr}\sqrt{M^\dagger M}$ is the trace norm. Note that, one can also define so called correlated catalytic transformations,
by replacing $\norm{\mu_{SC} - \sigma_S \otimes \tau_C}_1 < \varepsilon$ by $\norm{ \mu_{S} - \sigma_S}_1 < \varepsilon$ in Eq. (\ref{approximate}). Since $\norm{\mu_{SC} - \sigma_S \otimes \tau_C}_1 < \varepsilon$ implies  $\norm{\mu_{S} - \sigma_S}_1 < \varepsilon$, a transformation is achievable by approximate catalysis, then it is also achievable by correlated catalysis. Ref.~\cite{PhysRevX.8.041051} showed that when $\rho$ and $\sigma$ are incoherent, then $\rho$ can be transformed into $\sigma$ via thermal operations and approximate catalysis if and only if $S(\rho||\gamma_S)\geq S(\sigma||\gamma_S)$. Here, $S(\rho || \gamma) = \Tr \rho \pqty{\log \rho - \log \gamma}$ is the quantum relative entropy (or the quantum Kullback-Leibler (KL) divergence). This statement has been extended to general quantum states, but by replacing thermal operations with Gibbs-preserving operations~\cite{PhysRevLett.126.150502}. However as mentioned before, Gibbs-preserving operations are in general more powerful than thermal operations for transformation of general quantum states, due to the effects of coherence~\cite{faist2015gibbs}.

In this article, we will focus on studying the effects of coherence, which complements the existing results that ignore coherence~\cite{Horodecki_2013,PhysRevX.8.041051,PhysRevLett.126.150502}.
To do so, we analyze catalytic transformations in the resource theory of asymmetry using recently developed tools of marginal reducibility~\cite{ganardi2023catalytic,Ferrari2023,sapienza2019correlations}.

\emph{Catalytic transformations in resource theory of asymmetry.---}
In the resource theory of asymmetry, we associate a Hamiltonian $H$ to every system as in the resource theory of thermodynamics. The free states are given by states which are invariant under time-translation, i.e. states which commute with the hamiltonian $H$, while the free operations are the covariant operations~\cite{Marvian_2014}. As discussed earlier, we can think of restrictions arising in this theory as general limitations of coherence processing through thermal operations, since covariant operations can be implemented by thermal operations, when one has access to a source of free energy.

It is known that covariant operations and approximate catalysis cannot create coherence if the initial state is incoherent, which puts an important restriction on the power of approximate catalysis (no-broadcasting theorem)~\cite{PhysRevLett.123.020403,PhysRevLett.123.020404}. However, this does not cover the case when the initial state $\rho_S$  has non-zero coherence on some level pairs. We investigate exactly this setting, giving a sufficient condition for the existence of a catalytic transformations.

Following the formalism of Ref.~\cite{Takagi_2022}, we call $\mathcal{I}\pqty{\rho} = \Bqty{\Delta_{ij} = E_i - E_j\ |\ \bra{i} \rho \ket{j} \neq 0}$ the available coherences of $\rho$, i.e.\ all the energy differences for which $\rho$ has non-zero coherence.
We then construct the reachable coherences $\mathcal{J}\pqty{\rho} = \Bqty{\Delta E\ |\ \Delta E = \sum\limits_{\Delta_{ij} \in \mathcal{I}\pqty{\rho}} m_{ij} \Delta_{ij},\, m_{ij} \in \mathbb{Z} }$, which are the energy differences that can be written as integer multiples of those in $\mathcal{I}\pqty{\rho}$.
Intuitively, these are the coherences that are available in multiple copies of $\rho$~\cite{arxiv_Shiraishi_2023}.
Indeed, an explicit calculation shows that if $\rho_S$ contains non-zero coherence at energy difference $\Delta E_S$ and $\sigma_{S'}$ contains non-zero coherence at energy difference $\Delta E_{S'}$, then $\rho_S \otimes \sigma_{S'}$ contains non-zero coherence at energy difference $\Delta E_{S} + \Delta E_{S'}$. Let us now consider a qubit state $\ket{+_{ij}} = \frac{1}{\sqrt{2}} \pqty{\ket{i} + \ket{j}}$, where $\ket{i}$ and $\ket{j}$ are the energy eigenstates with eigenvalues given by $E_i$ and $E_j$ respectively, such that $\Delta_{ij}=E_i - E_j \in \mathcal{I}(\rho)$. From \cite{Takagi_2022}, we know that having enough number of such qubit states allows us to create $\rho$ via covariant operations. We formalise this as a Fact below.
\begin{fact}[{\cite[Proof of Theorem 2, step 3]{Takagi_2022}}]\label{fact:prepare}
    Given enough copies of $\Bqty{\ket{+_{ij}} \,|\, E_i - E_j \in \mathcal{I}(\rho) }$, we can create the state $\rho$ via covariant transformations.
\end{fact}
Ref. \cite{Takagi_2022} also shows that, there exists a sequence of correlated catalytic transformation, which transform $\rho$ into any number $\ket{+_{ij}}$, whenever $E_i - E_j \in \mathcal{J}\pqty{\rho}$.  The Fact below formalises this.

\begin{fact}[{\cite[Supplementary material, Proposition 10]{Takagi_2022}}]\label{fact:distill}
    There exists a sequence of correlated catalytic transformations (also called a quasi-correlated catalytic transformation in Ref.~\cite{Takagi_2022}) that create arbitrarily many copies of $\ket{+_{ij}}$, if $\Delta_{ij} \in \mathcal{J}\pqty{\rho}$.
\end{fact}
Combined, these two results shows that if $\mathcal{I}\pqty{\sigma} \subseteq \mathcal{J}\pqty{\rho}$, then there is a \emph{sequence} of correlated catalytic transformation that transform $\rho$ into $\sigma$.
However, it is not known whether these sequence of correlated catalytic transformations can be combined into a single correlated catalytic transformation.
In other words, we do not know whether correlated catalytic transformations form a transitive relation.
Indeed, Ref.~\cite{Takagi_2022} conjectured that if $\mathcal{I}\pqty{\sigma} \subseteq \mathcal{J}\pqty{\rho}$, then there is a correlated catalytic transformation that transforms $\rho$ into $\sigma$.
Since $\mathcal{I}\pqty{\sigma} \subseteq \mathcal{J}\pqty{\rho}$ is equivalent to $\mathcal{J}\pqty{\sigma} \subseteq \mathcal{J}\pqty{\rho}$, we can alternatively state the conjecture as $\mathcal{J}$ governs correlated catalytic transformations, ordered by set-inclusion.

In this work, we will show that these sequence of catalytic transformations \emph{can} in fact be combined. We will do this by relating the catalytic transformations to the notion of marginal reducibility.

\emph{Catalytic transformations and marginal reducibility.---} 
Ref.~\cite{ganardi2023catalytic} introduced the notion of marginal reducibility and showed its connection to catalytic transformations (Ref.~\cite{Ferrari2023,sapienza2019correlations} introduced the same notion in a different context).
A quantum state $\rho$ is marginally reducible into $\sigma$ if for any $\varepsilon, \delta > 0$ there exists a free operation $\Lambda$ and two natural numbers $m \leq n$ such that the following conditions hold for all $i \leq m$:
\begin{align} \label{eq:CorrelatedReducibility}
\Lambda\left(\rho^{\otimes n}\right) = \mu_{m},\,\,
\left\Vert \bqty{\mu_{m}}_i-\sigma\right\Vert _{1} <\varepsilon\,\,\text{and}\,\,
\frac{m}{n} \geq 1 - \delta. 
\end{align}
Here, $\mu_m$ is a quantum state of $m$ subsystems and the reduced state of $\mu_m$ on $i$-th subsystem is given by $\bqty{\mu_{m}}_i$.
 Compare this to correlated catalysis: we say that $\rho$ can be converted into $\sigma$ via correlated catalysis if for every $\varepsilon > 0$ there is a catalyst $\tau_C$ and a covariant operation $\Lambda_{SC}$ such that
\begin{equation} \label{eq:CorrelatedCatalysis}
\begin{aligned}\mu_{SC} =\Lambda_{SC}(\rho_{S}\otimes\tau_{C}),\,\,
||\mu_{S}-\sigma_{S}||_{1} < \varepsilon\,\,\text{and}\,\,\mu_{C}=\tau_{C}.
\end{aligned}
\end{equation}
Note that the notion of marginal reducibility involves many copies of initial and final states. On the other hand the definition of correlated catalysis only involves single copy of the initial (and final) state. Even though these concepts look different, it was shown in Ref.~\cite{ganardi2023catalytic} that these two concepts are deeply related. Let us first note that the works of \cite{PhysRevLett.127.080502,ganardi2023catalytic}, show us that if $\rho$ is marginally reducible into $\sigma$, then $\rho$ can be transformed into $\sigma$ via correlated catalysis.
\begin{fact}[{\cite[Proposition 3]{ganardi2023catalytic}}]\label{fact:mr-cc}
    Marginal reducibility implies correlated catalysis.
\end{fact}
In fact the converse also holds \cite{ganardi2023catalytic}, assuming that the initial state is \emph{distillable} i.e. if $\rho$ can be transformed into $\sigma$ via correlated catalysis and $\rho$ is distillable, then $\rho$ is marginally reducible into $\sigma$.
\begin{fact}[{\cite[Proposition 4]{ganardi2023catalytic}}]\label{fact:cc-mr}
    If the initial state is distillable, then correlated catalysis implies marginal reducibility.
\end{fact}
While these results are derived in the context of entanglement theory, they generalize to other resource theories with the appropriate modifications.
Let us now focus on what we mean by \emph{distillable}. A close reading of \cite{ganardi2023catalytic}, shows that the distillability condition can be relaxed to the following: Let $\rho$ can be transformed into $\sigma$ via correlated catalysis, we then say that $\rho$ is distillable, if many copies of $\rho$ can be used to create the catalyst with arbitrary accuracy. Formally speaking, if $\rho$ can be transformed into $\sigma$ via correlated catalysis, then for every $\varepsilon>0$, there exists a $\tau$ and a covariant operation $\Lambda_{SC}$, satisfying Eqs. (\ref{eq:CorrelatedCatalysis}). We then say $\rho$ is distillable if for every $\epsilon > 0$ and $\epsilon' > 0$ there is a number $k$ and a covariant operation $\Lambda'$ such that $\norm{\Lambda'\pqty{\rho^{\otimes k}} - \tau}_1 \leq \epsilon'$.

We first prove that marginal reducibility is a transitive relation.
\begin{lemma}\label{prop:transitivity}
    If $\rho$ is marginally reducible to $\mu$ and $\mu$ is marginally reducible to $\sigma$, then $\rho$ is marginally reducible to $\sigma$.
\end{lemma}
The argument relies on first principles analysis of marginal reducibility, and the detailed proof can be found in the Supplementary Material.
Along with the equivalence (between marginal reducibility and correlated catalysis), this lemma implies that correlated catalysis is transitive.
As a side remark, Lemma \ref{prop:transitivity} along with results in Ref.~\cite{ganardi2023catalytic} show that correlated catalytic transformations induce a transitive relation between distillable states in bipartite entanglement theory, which might be of independent interest.
This is due to the equivalence between marginal reducibility and correlated catalytic transformations for distillable states~\cite[Theorem 1]{ganardi2023catalytic}.

\emph{Main result.---}
Let us move to the main result of this Letter, namely the set-inclusion of reachable coherences implies that a catalytic covariant transformation is possible.
This suggests that thermodynamic restrictions that arise solely from coherence can be expressed simply as set-inclusion.
\begin{theorem}\label{th:main}
If $\mathcal{J}\pqty{\sigma} \subseteq \mathcal{J}\pqty{\rho}$, then $\rho$ can be transformed into $\sigma$ via approximate catalysis.
\end{theorem}
\begin{proof}
    Note that for any three states $\rho_1$, $\rho_2$ and $\rho_3$, if $\rho_1$ can be transformed into $\rho_2$ via approximate catalysis and $\rho_2$ can be transformed into $\rho_3$ via covariant operation, then $\rho_1$ can be transformed into $\rho_3$ via approximate catalysis.Therefore, because of Fact~\ref{fact:prepare}, it is enough to show that there exists an approximately catalytic operation transforming $\rho$ into arbitrarily many copies of $\ket{+_{ij}} = \frac{1}{\sqrt{2}} \pqty{\ket{i} + \ket{j}}$ for all $\Delta_{ij} \in \mathcal{J}\pqty{\rho}$.
    Since by assumption $\mathcal{J}\pqty{\sigma} \subseteq \mathcal{J}\pqty{\rho}$, this means we can obtain the state $\sigma$.
    
    Recall that there exists a sequence of correlated catalytic transformations that creates arbitrarily many copies of $\ket{+_{ij}}$ from $\rho$, if $\Delta_{ij} \in \mathcal{J}\pqty{\rho}$ (Fact~\ref{fact:distill}).
    We will show that each step in the sequence is also achievable in marginal asymptotics.
    Note that a close reading of the original proof (of \cite{Takagi_2022}) shows that every catalytic step in the sequence only deals with a two-level catalyst $\tau$ and initial state $\mu$, such that $\mathcal{I}\pqty{\tau} \subseteq \mathcal{I}\pqty{\rho}=\mathcal{I}\pqty{\mu}$. This condition ($\mathcal{I}\pqty{\tau} \subseteq \mathcal{I}\pqty{\mu}$), along with the fact that $\tau$ is a two-level system makes sure that $\mu$ is distillable \cite{marvian2020coherence,Takagi_2022}. Therefore correlated catalysis implies marginal reducibility for each step of the sequence (Fact~\ref{fact:cc-mr}).
    Then, using trasitivity of marginal reducibility (Lemma~\ref{prop:transitivity}), we conclude that the transformation as a whole is also achievable in marginal asymptotics.
    
    Using Fact~\ref{fact:mr-cc}, we conclude that there exists a correlated catalytic transformation that creates arbitrarily many copies of $\ket{+_{ij}}$ from $\rho$, if $\Delta_{ij} \in \mathcal{J}\pqty{\rho}$.
    Finally, since the target state is pure, correlated catalysis is equivalent to approximate catalysis~\cite[Proposition 8]{ganardi2023catalytic}, and the claim is shown.
\end{proof}

Our result answers the conjecture in Ref.~\cite{Takagi_2022} in the positive.
Note that Ref.~\cite{PhysRevA.103.022403} claimed the same result (see Theorem~2) for qubits, but Ref.~\cite{Takagi_2022} pointed out that the proof actually has a gap.
Here, we close the gap through a different technique.
This result implies that covariant operations can transform any two generic states via approximate catalysis. Furthermore, it provides some partial evidence for the following conjecture, that was phrased in a different way in Ref.~\cite{PhysRevA.103.022403}.
\begin{conjecture}
    There exists an approximately catalytic thermal operation transforming $\rho$ to $\sigma$ if and only if $S\pqty{\rho || \gamma} \geq S\pqty{\sigma || \gamma}$ and $\mathcal{J}\pqty{\sigma} \subseteq \mathcal{J}\pqty{\rho}$.
\end{conjecture}
In principle, this conjecture can be solved by proving these statements hold in the approximately catalytic setting:
(1) There exists a covariant operation transforming $\rho$ to $\sigma$ if and only if $\mathcal{J}\pqty{\sigma} \subseteq \mathcal{J}\pqty{\rho}$.
(2) The monotones governing thermal operations are exactly the combination of the monotones of Gibbs-preserving operations and covariant operations.
Since approximately catalytic transformations under Gibbs-preserving operations are allowed if and only if $S\pqty{\rho || \gamma} \geq S\pqty{\sigma || \gamma}$~\cite{PhysRevLett.126.150502}, these statements combined would solve the conjecture.
Recently there has been some partial results on proving these statements.
For example, the only if direction in statement (1) has been shown to hold in the special case when the frequencies of $\rho$ and $\sigma$ are related in a particular way~\cite{arxiv_Shiraishi_2023}.
Our result (and Ref.~\cite{arxiv_Shiraishi_2023}) prove the if direction in statement (1).
In addition, statement (2) is known to hold when the initial state is incoherent: if the target state is coherent, then the transformation is forbidden~\cite{PhysRevLett.123.020403}; otherwise it is determined by the free energy~\cite{PhysRevX.8.041051}.

We remark that our result implies that any correlation between the system and catalyst allows us to overcome most restriction from coherence.
This is in contrast to what happens if the catalyst must be completely uncorrelated to the system.
In this case, no amplification of coherence is possible, due to the additivity of quantum Fisher information on product states~\cite{PhysRevA.103.022403}.
This shows yet another example where allowing small correlations between the system and the catalyst significantly simplifies the transformation laws~\cite{PhysRevLett.127.150503,PhysRevLett.126.150502,PhysRevX.8.041051,PhysRevLett.130.240204}.

Let us turn to the properties of the catalyst that is needed in Theorem~\ref{th:main}.
Even though we relied heavily on the explicit construction of the catalyst in Ref.~\cite{Takagi_2022}, Theorem~\ref{th:main} does not provide any construction of the catalyst.
This is because in arguing transitivity of correlated catalysis through marginal reducibility, we are forced to deal with asymptotically many copies.
By analyzing how the error scales with the number of copies in the marginal asymptotic transformation, we can obtain an upper bound to the size of the catalyst that is needed in order to achieve a certain error (using the construction in Ref. \cite{PhysRevLett.127.080502,ganardi2023catalytic}).
However, this does not provide a complete picture as it was shown that there exist transformations that cannot be achieved by finite-dimensional catalysts in other resource theories~\cite{datta2022entanglement}.
With similar techniques, we show that the resource theory of asymmetry admits a related phenomenon where in general we need a catalyst whose Hamiltonian has an unbounded spectrum.
In particular, the transformation that amplifies coherence necessarily needs an unbounded catalyst.
The detailed argument, which uses the properties of quantum Fisher information~\cite{Marvian_2014,PhysRevA.94.012339}, can be found in the Supplementary Material.

\emph{Conclusions.---}
We have shown that allowing system-catalyst correlation in a catalytic procedure lifts most restrictions in the resource theory of asymmetry.
More precisely, Theorem~\ref{th:main} showed that as long as the initial state has some coherence in the right levels, we can amplify it and use it to prepare the target state.
Due to the connection of asymmetry with thermodynamics, this suggests that coherence only places a mild restriction on the allowed thermodynamic transformations.
We believe our results will provide a key step in formulating a fully general law of quantum thermodynamics.

\emph{Note.---}
During the completion of our manuscript, we became aware of an independent related work by [Naoto Shiraishi and Ryuji Takagi], submitted concurrently to the same arXiv posting~\cite{arxiv_Shiraishi_2023}.

\emph{Acknowledgements.---}This work was supported by the National Science Centre Poland (Grant No. 2022/46/E/ST2/00115) and within the QuantERA II Programme (No 2021/03/Y/ST2/00178, acronym ExTRaQT) that has received funding from the European Union's Horizon 2020 research and innovation programme under Grant Agreement No 101017733, and the ``Quantum Optical Technologies'' project, carried out within the International Research Agendas programme of the Foundation for Polish Science co-financed by the European Union under the European Regional Development Fund. The work of Tulja Varun Kondra is supported by the German Federal Ministry of Education and Research (BMBF) within the
funding program “quantum technologies – from basic re-
search to market” in the joint project QSolid (grant
number 13N16163).
\bibliography{main}
\section{SUPPLEMENTAL MATERIAL}
\section{Proof of Lemma 1}
Since $\rho$ is marginally reducible to $\mu$, by definition for any $\varepsilon_1, \delta_1 > 0$, there exist natural numbers $m_1, n_1$ and a free operation $\Lambda_1: S^{\otimes n_1} \rightarrow S^{\otimes m_1}$ such that
    \begin{align}
    \norm{\bqty{\Lambda_1\pqty{\rho^{\otimes n_1}}}_i - \mu}_1 < \varepsilon_1,
    \\
    \frac{m_1}{n_1} \geq 1 - \delta_1.
    \end{align}
    Similarly, for any $\varepsilon_2, \delta_2 > 0$, we can find $m_2, n_2$, and $\Lambda_2$ such that the corresponding statements holds for $\mu$ and $\sigma$.
    Now, let us take $\rho^{\otimes n_1 n_2}$ and split it into $n_2$ blocks, each containing $n_1$ copies of $\rho$ (see Figure~\ref{fig:transitivity}, vertical red blocks).
    When we apply $\Lambda_1$ to each block independently, the state of the system is $\pqty{\Lambda_1 \pqty{\rho^{\otimes n_1}}}^{\otimes n_2}$.
    Let us focus on the first subsystem of each block and denote the marginal as $\mu_{\varepsilon_1} = \bqty{\Lambda_1\pqty{\rho^{\otimes n_1}}}_1$.
    Note that if we take the all of the first subsystem from each block, we find that it is close to $\mu^{\otimes n_2}$
    \begin{align}
        \norm{
           \mu_{\varepsilon_1}^{\otimes n_2} - \mu^{\otimes n_2}
        }_1
        \leq
        n_2 \varepsilon_1.
    \end{align}
    This relation also holds for the second subsystem of each block, etc.
    Therefore, let us rearrange the subsystems into $m_1$ blocks of size $n_2$, where the $i$-th block contains all the $i$-th subsystem from the previous block (see Figure~\ref{fig:transitivity}, horizontal blue blocks).
    Suppose we apply $\Lambda_2$ in each block independently.
    On the first block, we will obtain
    \begin{align}
        \norm{
         \bqty{\Lambda_2 \pqty{ \mu_{\varepsilon_1}^{\otimes n_2}}}_j
         - \sigma
         }_1
         < n_2 \varepsilon_1 + \varepsilon_2,
    \end{align}
    for each subsystem, and analogously for all the other blocks.
    Therefore, for any $\varepsilon, \delta > 0$, we can set $\varepsilon_2 = \frac{\varepsilon}{2}, \varepsilon_1 = \frac{\varepsilon}{2 n_2}$ and $\delta_1 = \delta_2 = \frac{\delta}{2}$ and perform the protocol described above, showing that $\rho$ is marginally reducible to $\sigma$.

\section{Catalysts with unbounded Hamiltonian}
We will first show that the quantum Fisher information is non-increasing under any approximately catalytic transformations with a bounded catalyst Hamiltonian.
Formally, we have a bounded Hamiltonian by imposing that there exists a real number $M$ such that for any $\varepsilon > 0$, there exists a catalyst $\tau$ and a covariant operation $\Lambda$ such that Eq. (3) (of main text) is satisfied, and in addition $\norm{H_C}_{\infty} < M$ (here $H_C$ is the catalyst hamiltonian). Let us note that the quantum Fisher information of a state $\rho$ with Hamiltonian $H$ is given by
\begin{equation}
    F_Q \pqty{\rho, H}=2\sum_{k,l}\frac{\left(\lambda_k-\lambda_l\right)^2}{\lambda_k+\lambda_l}\left|\bra{\psi_k}H\ket{\psi_l}\right|^2.
\end{equation}
Here, $\{\ket{\psi_l}\}$ and $\{\lambda_l\}$ are the eigenvectors and eigenvalues of $\rho$ respectively. Recall that the quantum Fisher information is an asymmetry monotone that fulfills several properties~\cite{Marvian_2014,PhysRevA.94.012339}:
\begin{itemize}
\item (Monotonicity) For any covariant operation $\Lambda$, $F_Q \pqty{\rho, H} \geq F_Q \pqty{\Lambda(\rho), H}$.
\item (Additivity under tensor products) For any $\rho$ and $\tau$, we have $F_Q \pqty{\rho_S \otimes \tau_C , H_S\otimes I_{C}+I_S\otimes H_{C}} = F_Q \pqty{\rho_S, H_S} + F_Q \pqty{\tau_C, H_C}$.
\item (Continuity) For any $\rho$ and $\sigma$, we have $\abs{F_Q \pqty{\rho, H} - F_Q \pqty{\sigma, H}} \leq 32 \norm{H}_{\infty}^2 \sqrt{\norm{\rho - \sigma}_1}$.
\end{itemize}
Let us assume that $\rho$ can be transformed into $\sigma$ by approximate catalysis with a bounded Hamiltonian, and we define the quantities as in Eq. (3) (of main text). These three properties together implies that for any $\varepsilon > 0$, we have
\begin{align}
  F_Q \pqty{\rho_S, H_S}
  &=
    F_Q \pqty{\rho_S \otimes \tau_C, H_S\otimes I_{C}+I_S\otimes H_{C}}
    - F_Q \pqty{\tau_C, H_C}\nonumber
  \\
  &\geq
    F_Q \pqty{\mu_{SC}, H_S\otimes I_{C}+I_S\otimes H_{C}}
    - F_Q \pqty{\tau_C, H_C}\nonumber
  \\
  &\geq
    F_Q \pqty{\sigma_S \otimes \tau_C, H_S\otimes I_{C}+I_S\otimes H_{C}}\nonumber
    \\
    &- 32 \sqrt{\varepsilon} \norm{H_S\otimes I_{C}+I_S\otimes H_{C}}_{\infty}^2
    - F_Q \pqty{\tau_C, H_C} \nonumber
  \\
  &=
  F_Q \pqty{\sigma_S, H_S}
    - 32 \sqrt{\varepsilon} \norm{H_S\otimes I_{C}+I_S\otimes H_{C}}_{\infty}^2 \nonumber
  \\
  &\geq
    F_Q \pqty{\sigma_S, H_S}
    - 32 \sqrt{\varepsilon} \pqty{\norm{H_S}_{\infty} + M}^2.
\end{align}
Since these inequalities have to be satisfied for any $\varepsilon > 0$ and the norm $\norm{H_S}_{\infty}$ is not a function of $\varepsilon$, this means that $F_Q \pqty{\rho_S, H_S} \geq F_Q \pqty{\sigma_S, H_S}$, if we assume the Hamiltonian of the catalyst is bounded.
\begin{figure}
    \centering
    \includegraphics[width=.45\textwidth]{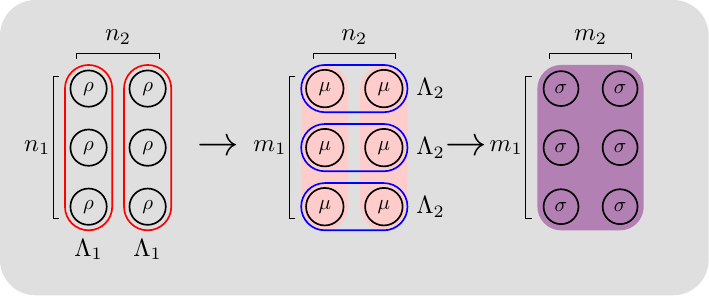}
    \caption{\label{fig:transitivity}
    Transitivity of marginal reducibility.
    Lemma 1 shows that we can compose the protocols for $\rho \to \mu$ and $\mu \to \sigma$ to obtain the protocol for $\rho \to \sigma$.
    }
\end{figure}

In contrast, Theorem 1 allows us to obtain arbitrarily many copies of $\frac{1}{\sqrt{2}} \pqty{\ket{0} + \ket{1}}$ from a single copy, showing that approximately catalytic transformations can increase the quantum Fisher information.
This can be explained by the fact that the product $\sqrt{\varepsilon} \norm{H_S\otimes I_{C}+I_S\otimes H_{C}}_{\infty}^2$ does not vanish in a general approximately catalytic transformation as $\varepsilon$ goes to zero.
This example shows that there are approximately catalytic transformations that are only attainable with an unbounded catalyst Hamiltonian.
\end{document}